\documentclass[prl,reprint,amsfonts,amsmath,amssymb,superscriptaddress,showpacs]{revtex4-1}

\usepackage{graphicx}
\usepackage{dcolumn}
\usepackage{bm}
\usepackage[bookmarks=true,colorlinks=true,linkcolor=black,citecolor=black,filecolor=black,menucolor=black,urlcolor=black]{hyperref}
\usepackage{amsthm}
\usepackage{mathrsfs}
\usepackage{comment}
\usepackage{color}

\newcommand{\bra}[1]{{\langle #1 |}}

\newcommand{\ket}[1]{{| #1 \rangle}}

\newcommand{\bracket}[1]{{\langle #1 \rangle}}

\DeclareMathOperator{\trace}{Tr}
\newcommand{\Trace}[1]{{\trace[ #1 ]}}

\DeclareMathOperator{\Ex}{\mathbb{E}}
\DeclareMathOperator{\Var}{Var}

\newcommand{\hilbert}{{\mathcal{H}}}

\newcommand{\RHO}{{\hat\rho}}

\newcommand{\LAMBDA}{{\hat{\bm{\lambda}}}}

\newcommand{\THeta}{{\bm{\theta}}}

\newcommand{\PVM}{{\bm{P}}}
\newcommand{\pvm}{{\hat P}}
\newcommand{\POVM}{{\bm{M}}}
\newcommand{\povm}{{\hat M}}

\newcommand{\est}{{\text{est}}}

\theoremstyle{definition}
\newtheorem{theorem}{Theorem}

\newcommand{\rnum}[1]{{\expandafter{\romannumeral #1}}}
\newcommand{\Rnum}[1]{{\uppercase\expandafter{\romannumeral #1}}}

\begin{document}

\preprint{APS/123-QED}

\title{Uncertainty Relation Revisited from Quantum Estimation Theory}
\author{Yu Watanabe}
\affiliation{%
  Department of Physics, University of Tokyo,
  7-3-1, Hongo, Bunkyo-ku, Tokyo 113-0033, Japan
}%
\author{Takahiro Sagawa}
\affiliation{%
  Department of Physics, University of Tokyo,
  7-3-1, Hongo, Bunkyo-ku, Tokyo 113-0033, Japan
}%
\author{Masahito Ueda}
\affiliation{%
  Department of Physics, University of Tokyo,
  7-3-1, Hongo, Bunkyo-ku, Tokyo 113-0033, Japan
}%
\affiliation{%
  ERATO Macroscopic Quantum Control Project, JST, 
  2-11-16 Yayoi, Bunkyo-ku, Tokyo 113-8656, Japan
}%

\date{\today}

\begin{abstract}
  By invoking quantum estimation theory we formulate bounds of errors in quantum measurement for arbitrary quantum states and observables in a finite-dimensional Hilbert space.
  We prove that the measurement errors of two observables satisfy Heisenberg's uncertainty relation,
  find the attainable bound, and provide a strategy to achieve it.
\end{abstract}


\maketitle

Quantum theory features two types of uncertainty: indeterminacy of observables and complementarity of quantum measurements.
The indeterminacy~\cite{bib:indeterminacy} reflects the inherent nature of a quantum system 
alone~\cite{bib:kennard,bib:robertson,bib:entropic-uncertainty}.
On the other hand, the complementarity~\cite{bib:complementarity} involves quantum measurement, 
and the estimation of the quantum state from the measurement outcomes is essential~\cite{bib:helstrom-1,bib:helstrom-2,bib:holevo}.
However, how to optimize the measurement and estimation for a given quantum system has remained an outstanding issue.
The purpose of this Letter is to report the resolution to this problem.

The complementarity implies that we cannot simultaneously perform precise measurements of non-commutable observables.
There must exist trade-off relations of the measurement errors concerning non-commutable observables.
Whereas a number of trade-off relations have been found, 
they are neither attainable for all quantum states and observables~\cite{bib:holevo,bib:arthurs-kelly-goodman,bib:ozawa,bib:yuen-lax,bib:nagaoka} 
nor applicable for all quantum systems~\cite{bib:sagawa,bib:englert,bib:busch}.
Due to the advances in controlling quantum states, it is now possible to implement
a scheme that performs a projection measurement on a part of samples and another projection measurement 
on the other samples~\cite{bib:cold-atom-theory,bib:cold-atom-experiment-1,bib:cold-atom-experiment-2,bib:cold-atom-experiment-3}.
However, the attainable bound of the measurement errors for such a scheme is yet to be clarified.

\begin{figure}[hbt]
  \centering
  \includegraphics[width=246pt]{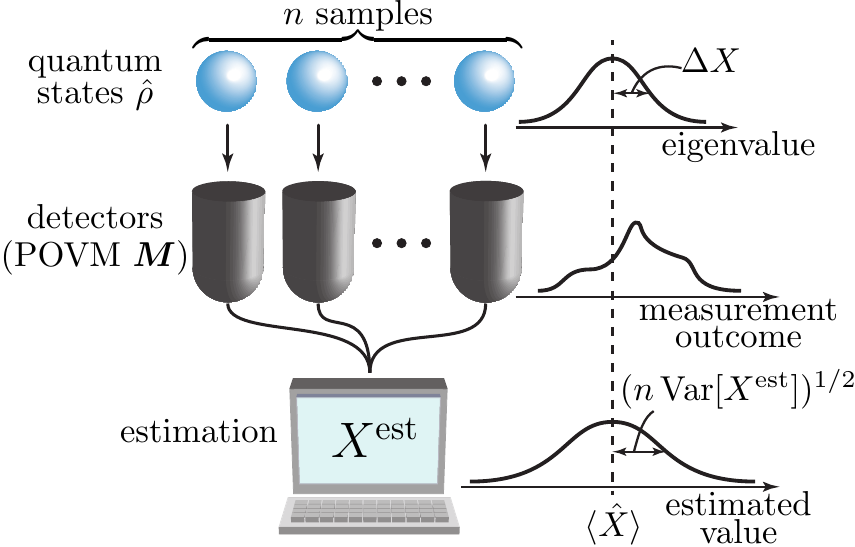}
  \caption{(Color online)
    A measurement described by POVM $\POVM$ is performed to retrieve the information about $\bracket{\hat X}$.
    Since the distribution of the measurement outcomes is, in general, incompatible with $\bracket{\hat X}$ due to the noise introduced in the measurement process,
    it is necessary to estimate $\bracket{\hat X}$ from the outcomes.
    The distribution of the estimated values becomes broader than the original distribution for $\RHO$ due to the measurement errors.
  }
  \label{fig:flowchart}
\end{figure}

In this Letter, we obtain the following three results.
Firstly, we prove that for all measurements the measurement errors of non-commutable observables are bounded from below by the commutation relation concerning observables.
This implies that not only quantum fluctuations but also the measurement errors are bounded by the commutation relation.
However, the bound cannot be achieved for all quantum states and observables.
Secondly, we find the attainable bound for the measurements that perform a projection measurement with or without noise on a part of samples and another measurement on the rest.
We propose a scheme of the experimental setup that achieves the bound.
Thirdly, we numerically vindicate that the attainable bound cannot be violated for all measurements.
Therefore, we conjecture that all measurements satisfy the attainable trade-off relation,
and that the measurements that achieve the attainable bound are optimal for obtaining information concerning two non-commutable observables.

Conventionally, the complementarity in quantum measurement is discussed in terms of the variance of the measurement outcomes~\cite{bib:arthurs-kelly-goodman,bib:ozawa}.
However, the variance of the measurement outcomes per se does not always give a quantitative error concerning the measurement.
To quantify this error in the measurement, it is essential to invoke quantum estimation theory (see Fig.~\ref{fig:flowchart}).
The measurement error is quantified by the difference between the information obtained by the measurement and one by the precise measurement concerning the observable.
The information content corresponding to quantum estimation theory is the Fisher information that gives precision of the estimated value calculated from the measurement outcomes.
However, it is challenging to find the attainable bound for the Fisher information.
Several bounds of the uncertainty relation has been derived by using the Fisher information~\cite{bib:yuen-lax,bib:holevo,bib:nagaoka}, but they are not attainable.
The crucial point of our successful finding of the attainable bound is that we express the relevant operators in terms of the generators of the Lie algebra $\mathfrak{su}(d)$, 
where $d$ is the dimension concerning Hilbert space $\hilbert$.
This greatly facilitates the analysis of our results.

Given $n$ independent and identically distributed (i.i.d.) unknown quantum states $\RHO$,
we perform the same measurement on each of them.
If we want to know the expectation value $\bracket{\hat X}:=\Trace{\RHO\hat X}$ of a single observable $\hat X$,
the optimal strategy is to perform the projection measurement $\PVM = \{\pvm_i\}$ corresponding to the spectral decomposition of $\hat X = \sum_i \alpha_i \pvm_i$ and 
then to calculate the estimated value of $\bracket{\hat X}$ as
$X^\est = \sum_i \alpha_i n_i / n$, where $\alpha_i$ is the $i$th eigenvalue of $\hat X$ and $n_i$ is the number of times the outcome $i$ is obtained.
The variance of the estimator $X^\est$, which quantifies the estimation error, is calculated to be 
$\Var[X^\est] := \Ex[(X^\est)^2] - \Ex[X^\est]^2 = (\Delta X)^2 / n$,
where $(\Delta X)^2 := \bracket{\hat X^2} - \bracket{\hat X}^2$, and $\Ex[X^\est]$ is the expectation value of $X^\est$ taken over the probability that the outcome $i$ is obtained $n_i$ times: 
$p(n_1, n_2, \dots) = n!\prod_i p_i^{n_i}/n_i!$ with $p_i = \Trace{\RHO\pvm_i}$.

When we perform the positive operator-valued measure (POVM) measurement $\POVM=\{\povm_i\}$, 
the variance of the estimator is asymptotically greater than that of the optimal one:
$\lim_{n\rightarrow\infty}n\Var[X^\est] \ge (\Delta X)^2$ (see Fig.~\ref{fig:flowchart}), 
where the left-hand side (LHS) and the right-hand side (RHS) show the variance of the concerned estimator and that of the optimal case per sample, respectively.
(Note that $\Var[X^\est]$ decreases as $n^{-1}$.)
The difference between them is caused by the measurement error of $\hat X$ for $\POVM$.
To quantify the error in the measurement, it is necessary to use the estimator that minimizes the variance.
We define the measurement error as
\begin{equation}
  \varepsilon(\hat X; \POVM) := \min_{X^\est} \lim_{n\rightarrow\infty} n\Var[X^\est] - (\Delta X)^2,
\end{equation}
where the minimization is taken over all so-called consistent estimators that asymptotically converge to $\bracket{\hat X}$: 
$\lim_{n\rightarrow\infty} \text{Prob}(|X^\est - \bracket{\hat X}| < \delta) = 1$
for all quantum states $\RHO$ and an arbitrary $\delta > 0$.
Examples of the consistent estimator include the average of eigenvalues, $\sum_i \alpha_i n_i/n$, for the projection measurement, 
and the maximum likelihood estimator for the POVM measurement.
These examples also minimize $\lim_{n\rightarrow\infty} n\Var[X^\est]$.
As shown later, $\varepsilon(\hat X; \POVM)$ can be expressed in terms of the Fisher information.

The first result in this Letter is that we prove in Theorem~\ref{theo:heisenberg} below that the measurement errors of two observables $\hat X_1$ and $\hat X_2$ satisfy 
\begin{equation}
  \varepsilon(\hat X_1; \POVM) \varepsilon(\hat X_2; \POVM) \ge \frac{1}{4}\left|\bracket{[\hat X_1, \hat X_2]}\right|^2, \label{eq:heisenberg}
\end{equation}
where the square brackets $[\ ,\ ]$ denote the commutator.
Heisenberg originally discussed the trade-off relation between the measurement error of an observable and the disturbance to another non-commutable observable caused by the measurement.
From this argument, it can be expected that the trade-off relation between measurement errors exists.
We have proved this in the form of \eqref{eq:heisenberg}.
Holevo proved a similar formula for position and momentum for the coherent state~\cite{bib:holevo}.
Equation \eqref{eq:heisenberg} is satisfied for all quantum states and observables on any finite dimensional Hilbert space.
However, the equality in \eqref{eq:heisenberg} is not achievable for all quantum states and observables
(see the dash-dotted curve in Fig.~\ref{fig}(a)).

A measurement scheme that performs a projection measurement $\PVM_1$ on $n_1$ samples and another projection $\PVM_2$ on $n_2 = n - n_1$ samples
is asymptotically equivalent to the POVM measurement that randomly performs 
those two projection measurements with probabilities $q_\nu = n_\nu/n$ ($\nu = 1,2$).
We define a set of such random projection measurements as $\mathcal{M}_{\text{random}} := \{q_1\PVM_1 + q_2\PVM_2\,|\,\PVM_1, \PVM_2\in\mathcal{P},\,q_1, q_2 \ge 0,\,q_1 + q_2 = 1\}$,
where $\mathcal{P}$ denotes the entire set of projection measurements, and $q_1\PVM_1 + q_2\PVM_2 = \{q_\nu\pvm_{\nu,i}\}$ for $\PVM_\nu = \{\pvm_{\nu,i}\}$.
In real experimental setups, measurements always suffer from noises. A typical noise model causes
a loss of the visibility for a projection measurement $\PVM$. Such a noisy measurement can be expressed as $\POVM = F\PVM = \{\sum_j F_{ij}\pvm_j\}$,
where the matrix $F$ is the so-called information proccesing matrix whose elements satisfy $F_{ij} \ge 0$ and $\sum_j F_{ij} = 1$.
The class of measurements described by $F\PVM$ include a broad class of experimentally realizable measurements.
For example, a typical scheme of the quantum non-demolition (QND) measurement belongs to this class~\cite{bib:cold-atom-theory,bib:cold-atom-experiment-1,bib:cold-atom-experiment-2}.
We note that the noise of a measurement in this class is equivalent to a classical noise that is characterized by a classical noisy channel with $F_{ij}$.
We define a set of measurements $\mathcal{M}_{\text{noisy}} := \{F\POVM\,|\,\POVM\in\mathcal{M}_{\text{random}}, F_{ij} \ge 0, \sum_j F_{ij} = 1\}$, 
which include random measurements consisting of noisy projection measurements.
Note that those sets of measurements satisfy $\mathcal{M}_{\text{random}} \subset \mathcal{M}_{\text{noisy}} \subset \mathcal{M}_{\text{all}}$,
where $\mathcal{M}_{\text{all}}$ denotes the totality of POVM measurements.

The second result in this Letter is that we prove in Theorem~\ref{theo:achievability} that the attainable bound of the measurement errors for $\mathcal{M}_{\text{noisy}}$ is
\begin{align}
  &\varepsilon(\hat X_1; \POVM) \varepsilon(\hat X_2; \POVM) \notag \\
  &\qquad \ge (\Delta_Q X_1)^2(\Delta_Q X_2)^2 - [\mathcal{C}_Q(\hat X_1, \hat X_2)]^2. \label{eq:our-uncertainty}
\end{align}
The third result is that we numerically vindicate that all $\POVM\in\mathcal{M}_{\text{all}}$ satisfy inequality \eqref{eq:our-uncertainty}.
We rigorously prove this for the qubit system ($\dim\hilbert = 2$)~\cite{bib:supplemental-material}.
Here $\Delta_Q$ and $\mathcal{C}_Q$ are defined as follows.
Let $\hilbert_a$ ($a=A,B,\dots$) be the simultaneous irreducible invariant subspace of $\hat X_\mu$, and $\pvm_a$ the projection operator on $\hilbert_a$.
We define the probability distribution as $p_a:=\bracket{\pvm_a}$ and the post-measurement state of the projection measurement $\{\pvm_A,\pvm_B,\dots\}$ as $\RHO_a := \pvm_a \RHO \pvm_a / p_a$.
Then, $\Delta_Q$ and $\mathcal{C}_Q$ are defined as
$(\Delta_Q X_\mu)^2:=\sum_a p_a (\Delta_a X_\mu)^2$,
and $\mathcal{C}_Q(\hat X_1, \hat X_2) := \sum_a p_a \mathcal{C}_{s,a} (\hat X_1, \hat X_2)$,
where $(\Delta_a X_\mu)^2$ and 
$\mathcal{C}_{s,a}(\hat X_1, \hat X_2) := \frac{1}{2}\Trace{\RHO_a\{\hat X_1, \hat X_2\}} - \Trace{\RHO_a\hat X_1}\Trace{\RHO_a\hat X_2}$ 
are the variance and the symmetrized correlation for $\RHO_a$, respectively,
and the curly brackets $\{\ ,\ \}$ denote the anti-commutator.

If $\hat X_1$ and $\hat X_2$ are simultaneously block-diagonalizable,
then quantum fluctuations and correlations of observables are determined by the diagonal blocks of $\RHO$.
(Note that $\bracket{\hat X_\mu}$ is independent of the off-diagonal blocks of $\RHO$.)
If two observables are commutable with each other, the RHS of \eqref{eq:our-uncertainty} vanishes.

Inequality \eqref{eq:our-uncertainty} is stronger than \eqref{eq:heisenberg} and the trade-off relations 
obtained by Nagaoka~\cite{bib:nagaoka} (see Fig.~\ref{fig}(a)),
and reduces to the trade-off relation found in Ref.~\cite{bib:sagawa} for $\dim\hilbert = 2$ and $\RHO=\hat I/2$.
The optimal measurement of Englert's complementarity~\cite{bib:englert} for $\dim\hilbert=2$ achieves the bound set by \eqref{eq:our-uncertainty}.

\begin{figure}[t]
  \centering
  \includegraphics[width=246pt]{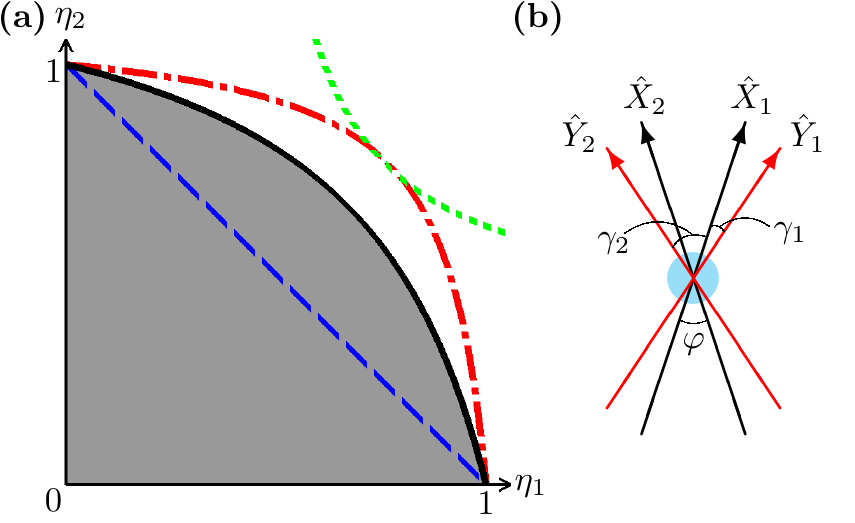}
  \caption{(Color online)
    \textbf{(a)}Plots of measurement errors of $10^9$ randomly chosen POVMs for $\dim\hilbert=4$ ($S=3/2$), $\RHO=\{\hat I/(2S+1) + \ket{S}\bra{S}\}/2$, 
    $\hat X_1 = \hat S_x$ and $\hat X_2 = (\sqrt{3}\hat S_x + \hat S_y)/2$,
    where $\eta_\mu := [\varepsilon(\hat X_\mu;\POVM)/(\Delta X_\mu)^2 + 1]^{-1}$ with $0\le \eta_\mu \le 1$.
    The red dash-dotted, black solid, blue dashed, and green dotted curves show the bounds set by \eqref{eq:heisenberg}, \eqref{eq:our-uncertainty}, \eqref{eq:random-measurement}, 
    and the inequality obtained in Ref.~\cite{bib:nagaoka}, respectively.
    \textbf{(b)}The directions of the spin observables $\hat Y_1$ and $\hat Y_2$ that attain the bound set by \eqref{eq:our-uncertainty} for two spin observables $\hat X_1$ and $\hat X_2$.
    The projection measurement of $\hat Y_\nu$ can be implemented, for example, by using cold atoms and linearly polarized laser whose propagation direction is specified by $\gamma_\nu$.
  }
  \label{fig}
\end{figure}

A simplest but not optimal way to estimate $\bracket{\hat X_1}$ and $\bracket{\hat X_2}$ is 
to perform the projection measurement of $\hat X_\mu$ with probability $q_\mu$ ($\mu = 1, 2$) with $q_1 + q_2 = 1$.
The measurement errors in this measurement satisfy
\begin{equation}
  \varepsilon(\hat X_1;\POVM)\varepsilon(\hat X_2;\POVM) = (\Delta_Q X_1)^2 (\Delta_Q X_2)^2. \label{eq:random-measurement}
\end{equation}
On the other hand, as shown in Theorem~\ref{theo:achievability}, the bound set by \eqref{eq:our-uncertainty} is attained by 
probabilistically performing the projection measurements of two observables $\hat Y_1$ and $\hat Y_2$,
where $\hat Y_1$ and $\hat Y_2$ are described by certain linear combinations of $\hat X_1$ and $\hat X_2$.
Since this optimal measurement can utilize the correlation of $\hat X_1$ and $\hat X_2$, the bound is expressed as \eqref{eq:our-uncertainty}.

We emphasize that the bound set by \eqref{eq:our-uncertainty} can be achieved for all quantum states and observables, 
whereas the bound set by \eqref{eq:heisenberg} cannot.
For example, for $\RHO=r\hat I/(2S+1) + (1-r)\ket{S}\bra{S}$, $\hat X_1 = \hat S_x$, $\hat X_2 = \hat S_x\cos\varphi + \hat S_y\sin\varphi$, and $q_1 = q_2 = 1/2$,
the measured observable $\hat Y_\nu = \hat S_x\cos\gamma_\nu + \hat S_y\sin\gamma_\nu$ is determined by 
the solution to $\cos\varphi\cos^2(\gamma_1 - \gamma_2) - 2\cos(\gamma_1 + \gamma_2 - \varphi)\cos(\gamma_1 - \gamma_2) + \cos\varphi = 0$,
where $\hat I$ is the identity operator, $\hat S_i$ is the spin operator of total spin $S$ in the $i$ $(=x,y,z)$ direction,
and $\ket{S}$ is the eigenstate of $\hat S_z$ with eigenvalue $S$.
The RHS of \eqref{eq:heisenberg} and that of \eqref{eq:our-uncertainty} 
are given by $[\frac{1}{2}(1-r)S\sin\varphi]^2$ and $[rS(2S-1)/6 + S/2]^2\sin^2\varphi$, respectively.
Such an optimal measurement can be implemented, for example, by using cold atoms~\cite{bib:cold-atom-theory,bib:cold-atom-experiment-1,bib:cold-atom-experiment-2,bib:cold-atom-experiment-3}.
By letting an ensemble of atoms interact with a linearly polarized off-resonant laser whose propagation direction is parallel to that specified by $\gamma_\nu$ in $\hat Y_\nu$ (see Fig.\ref{fig}(b)),
the angle of the paramagnetic Faraday rotation of the laser polarization carries information about $\bracket{\hat Y_\nu}$.
The rotation angle can be detected by a polarimeter using a polarization-dependent beam splitter.
If the intensity of the laser is sufficiently strong, this scheme achieves the projection measurement of $\hat Y_\nu$.

Our trade-off relation \eqref{eq:our-uncertainty} is rigorously proved for the measurements in $\mathcal{M}_{\text{all}}$ for $\dim\hilbert=2$
and $\mathcal{M}_{\text{noisy}}$ for $\dim\hilbert \ge 3$.
For a higher dimensional Hilbert space from $\dim\hilbert=3$ to $7$~\cite{bib:number}, we numerically calculate the measurement errors of $10^9$ randomly chosen POVMs in $\mathcal{M}_{\text{all}}$
for randomly chosen $10$ pairs of quantum states and two observables $(\RHO, \hat X_1, \hat X_2)$.
We find that the calculated measurement errors satisfy \eqref{eq:our-uncertainty}.
A typical example of the numerical calculation is shown in Fig.~\ref{fig}(a).
The area within the bound is blacked out by $10^9$ data points with no point found outside of the bound.
Therefore, we conjecture that \eqref{eq:our-uncertainty} is satisfied not only for $\mathcal{M}_{\text{noisy}}$ but also for $\mathcal{M}_{\text{all}}$.

To prove the theorems stated below, we introduce the decomposition of the Hermitian operators on the $d$-dimensional Hilbert space 
by the generators of the Lie algebra $\mathfrak{su}(d)$.
The generators $\LAMBDA = \{\hat\lambda_1,\dots,\hat\lambda_{d^2-1}\}$ are traceless and Hermitian, satisfying
$\Trace{\hat\lambda_i\hat\lambda_j} = \delta_{ij}$.
The quantum state $\RHO$ can be expressed as $\RHO=d^{-1}\hat I + \THeta\cdot\LAMBDA$, 
where $\THeta\in\mathbb{R}^{d^2-1}$ is an $(d^2-1)$-dimensional real vector, 
and $\THeta\cdot\LAMBDA=\sum_{i=1}^{d^2-1}\theta_i\hat\lambda_i$.
An arbitrary observable can also be expanded in terms of the same set of generators as $\hat X = x_0\hat I + \bm{x}\cdot\LAMBDA$.
Then, the expectation value can be written as $\bracket{\hat X} = x_0 + \bm{x}\cdot\THeta$.
Therefore, estimating $\bracket{\hat X}$ amounts to estimating $\bm{x}\cdot\THeta$.
For any consistent estimator $X^\est$ of $\bracket{\hat X}$, the variance $\Var[X^\est]$ satisfies the following Cram\'er-Rao inequality~\cite{bib:cramer-rao}:
$\lim_{n\rightarrow\infty}n\Var[X^\est] \ge \bm{x}\cdot [J(\POVM)]^{-1}\bm{x}$,
where $J(\POVM)$ is the Fisher information matrix whose $ij$ element is defined as 
$[J(\POVM)]_{ij} := \sum_k p_k(\partial_i\log p_k)(\partial_j\log p_k)$, where $\partial_i = \partial/\partial \theta_i$.
For all quantum states and POVMs, there exists some estimator, for example, the maximum likelihood estimator, 
that achieves the equality of the Cram\'er-Rao inequality.

The matrix $J(\POVM)$ varies with varying the POVM, but it is bounded from above by
the quantum Cram\'er-Rao inequality~\cite{bib:caves:quantum-cramer-rao}: $J(\POVM) \le J_Q$,
where $J_Q$ is the quantum Fisher information matrix~\cite{bib:helstrom:quantum-fisher}
which is a monotone metric on the quantum state space with the coordinate system $\THeta$.
The quantum Fisher information matrix is not uniquely determined, but from the monotonicity there exist the minimum $J_S$ and the maximum $J_R$~\cite{bib:petz},
where $J_S$ ($J_R$) is the symmetric (right) logarithmic derivative Fisher information matrix.
Their $ij$ elements are defined as $[J_S]_{ij} := \frac{1}{2}\bracket{\{\hat L_i, \hat L_j\}}$ 
and $[J_R]_{ij} := \bracket{\hat L'_j \hat L'_i}$,
where $\hat L_i$ and $\hat L_i'$ are defined to be the solution to 
$\partial_i \RHO = \frac{1}{2}\{\RHO, \hat L_i\}$ and $\partial_i \RHO = \RHO\hat L'_i$.
It can be shown that $[J_S^{-1}]_{ij} = \mathcal{C}_s(\hat\lambda_i, \hat\lambda_j)$ and $[J_R^{-1}]_{ij} = \mathcal{C}(\hat\lambda_i, \hat\lambda_j)$, 
where $\mathcal{C}(\hat X, \hat Y) := \bracket{\hat X\hat Y} - \bracket{\hat X}\bracket{\hat Y}$
is the (non-symmetrized) correlation function of two observables.
From the Cram\'er-Rao inequality and $\bm{x}\cdot J_S^{-1}\bm{x} = \bm{x}\cdot J_R^{-1}\bm{x} = (\Delta X)^2$, 
we find that $\varepsilon(\hat X; \POVM) = \bm{x}\cdot[J(\POVM)^{-1} - J_Q^{-1}]\bm{x} \ge 0$
is satisfied for any quantum Fisher information.

We now prove the following theorems.
\begin{theorem}
  \label{theo:heisenberg}
  For all quantum states $\RHO$ and observables $\hat X_\mu$, any POVM $\POVM\in\mathcal{M}_{\text{all}}$ satisfies \eqref{eq:heisenberg}.
\end{theorem}
\begin{proof}
  From the quantum Cram\'er-Rao inequality, $J(\POVM)^{-1} - J_R^{-1} \ge 0$.
  Since $J(\POVM)$ is real symmetric and $J_R$ is Hermitian, 
  for all observables $\hat X_\mu = x_{0,\mu}\hat I + \bm{x}_\mu\cdot\LAMBDA$ and $k\in\mathbb{R}$,
  $(\bm{x}_1 + i k\bm{x}_2)^\dagger[J(\POVM)^{-1} - J_R^{-1}](\bm{x}_1 + i k\bm{x}_2) \ge 0$.
  Since the discriminant of the quadratic polynomial in the LHS is always negative, \eqref{eq:heisenberg} is proved.
\end{proof}

\begin{theorem}
  \label{theo:achievability}
  For all quantum states $\RHO$ and observables $\hat X_\mu$, any POVM $\POVM\in\mathcal{M}_{\text{noisy}}$ satisfies \eqref{eq:our-uncertainty}.
  Moreover, the measurements that achieve the equality of \eqref{eq:our-uncertainty} exist for all quantum states and observables.
\end{theorem}
\begin{proof}
  If two POVMs satisfy $\POVM' = F\POVM$ with an information processing matrix $F$, those POVMs satisfy $J(\POVM') \le J(\POVM)$.
  Hence, we have only to consider the case when $\POVM \in \mathcal{M}_{\text{random}}$.

  Let $\hat Y_\nu = y_{\nu,0}\hat I + \bm{y}_\nu\cdot\LAMBDA$ ($\nu = 1,2$) be a linear combination of $\hat X_\mu = x_{\mu,0}\hat I + \bm{x}_\mu\cdot\LAMBDA$, 
  and $A = (a_{\mu\nu})$ be its coefficient: $\hat X_\mu = \sum_\nu a_{\mu\nu}\hat Y_\nu$.
  We consider the POVM measurement $\POVM = q_1\PVM_1 + q_2\PVM_2 \in\mathcal{M}_{\text{random}}$, 
  where $\PVM_\nu$ corresponds to the spectral decompositions of the observables $\hat Y_\nu = \sum_i \beta_{\nu,i}\pvm_{\nu,i}$.
  The inverse of $J(\POVM)$ can be obtained as 
  $\bm{y}_\nu\cdot J(\POVM)^{-1}\bm{y}_\nu = (\Delta Y_\nu)^2 + (q_\nu^{-1} - 1)(\Delta_Q Y_\nu)^2$, 
  and $\bm{y}_1\cdot J(\POVM)^{-1}\bm{y}_2 = \mathcal{C}_s(\hat Y_1, \hat Y_2) - \mathcal{C}_Q(\hat Y_1, \hat Y_2)$.
  Therefore,
  \begin{align}
    &\varepsilon(\hat X_1; \POVM)\varepsilon(\hat X_2; \POVM)
    \notag \\
    & \quad\ge \det
    \begin{pmatrix}
      \frac{q_2}{q_1}(\Delta_Q Y_1)^2 & -\mathcal{C}_Q(\hat Y_1, \hat Y_2) \\
      -\mathcal{C}_Q(\hat Y_1, \hat Y_2) & \frac{q_1}{q_2}(\Delta_Q Y_2)^2
    \end{pmatrix}
    (\det A)^2 \notag \\
    &\quad = (\Delta_Q X_1)^2(\Delta_Q X_2)^2 - [\mathcal{C}_Q(\hat X_1, \hat X_2)]^2.
  \end{align}
  The condition for the equality to hold is that $\bm{x}_1\cdot(J(\POVM)^{-1} - J_S^{-1})\bm{x}_2$ vanishes. 
  The observables $\hat Y_\nu$ that satisfy this condition exist for all $\RHO$.
\end{proof}

To summarize, we formulate the complementarity of the quantum measurement in a finite-dimensional Hilbert space by invoking quantum estimation theory.
To quantify the information retrieved by the measurement, it is essential to take into account the estimation process.
We prove that the measurement errors of non-commutable observables satisfy Heisenberg's uncertainty relation,
and find the stronger bound that can be achieved for all quantum states and observables.

\begin{acknowledgments}
This work was supported by 
Grants-in Aid for Scientific Research (KAKENHI 22340114 and 22103005),
the Global COE Program ``the Physical Sciences Frontier'',
and the Photon Frontier Network Program, from MEXT of Japan.
Y.W. and T.S. acknowledge support from JSPS (Grant No. 216681 and No. 208038, respectively).
\end{acknowledgments}


\begin{thebibliography}{99}
\bibitem{bib:number}
  This range includes prime numbers $(3, 5, 7)$, a power of prime $(4)$, and a composite number that is not a power of prime $(6)$.
\end{thebibliography}
\end{document}